\documentclass[letterpaper, 10 pt, conference]{ieeeconf}  

\IEEEoverridecommandlockouts                              

\usepackage[utf8]{inputenc}     

\usepackage[T1]{fontenc}        

\usepackage[english]{babel}              

\usepackage{graphicx}           
\usepackage{amsmath,amssymb}    
\usepackage{color}
\usepackage[usenames]{xcolor}   
\usepackage{microtype} 

\usepackage{marginnote}         

\usepackage{mathtools}
\usepackage{siunitx}
\usepackage{pbox}
\usepackage{dsfont}
\usepackage{algorithm}

\usepackage{tikz}				
\usetikzlibrary{arrows}				
\usetikzlibrary{shapes.misc}
\usetikzlibrary{patterns}
\usepackage{pgfplots}
\usetikzlibrary{calc}

\colorlet{istorange}{orange}
\colorlet{istgreen}{green!50!black}
\colorlet{istblue}{blue} 
\colorlet{istred}{red!90!black}

\newtheorem{theorem}{Theorem}
\newtheorem{lemma}{Lemma}
\newtheorem{corollary}{Corollary}
\newtheorem{assumption}{Assumption}
\newtheorem{definition}{Definition}
\newtheorem{proposition}{Proposition}
\newtheorem{remark}{Remark}

\newcommand{\I}{\mathbb{I}}
\newcommand{\R}{\mathbb{R}}

\newcommand{\N}{\mathbb{N}}

\makeatletter
\renewcommand*{\@opargbegintheorem}[3]{\trivlist
	\item[\hskip \labelsep{\textit{ #1\ #2}}] \textit{(#3)}: }
\makeatother

\title{\LARGE \bf Improved stability conditions for systems under aperiodic sampling: model- and data-based analysis$^*$}

\author{Stefan Wildhagen, Julian Berberich, Matthias Hirche and Frank Allg{\"o}wer
	\thanks{$^*$Funded by Deutsche Forschungsgemeinschaft (DFG, German Research Foundation) under Germany's Excellence Strategy - EXC 2075 - 390740016 and under grant AL 316/13-2 - 285825138. We acknowledge the support by the Stuttgart Center for Simulation Science (SimTech) and thank the International Max Planck Research School for Intelligent Systems (IMPRS-IS) for supporting J. Berberich. The authors are with the University of Stuttgart, Institute for Systems Theory and Automatic Control, Germany. {\tt\small \{wildhagen,}{\tt\small berberich,} {\tt\small hirche,}{\tt\small allgower\}@ist.uni-}{\tt\small stuttgart.de}} %
}

\begin{document}

\maketitle
\thispagestyle{empty}
\pagestyle{empty}

\begin{abstract}
Discrete-time systems under aperiodic sampling may serve as a modeling abstraction for a multitude of problems arising in cyber-physical and networked control systems. Recently, model- and data-based stability conditions for such systems were obtained by rewriting them as an interconnection of a linear time-invariant system and a delay operator, and subsequently, performing a robust stability analysis using a known bound on the gain of this operator. In this paper, we refine this approach: First, we show that the delay operator is input-feedforward passive and second, we compute its gain exactly. Based on these findings, we derive improved stability conditions both in case of full model knowledge and in case only data are available. In the latter, we require only a finite-length and potentially noisy state-input trajectory of the unknown system. In two examples, we illustrate the reduced conservativeness of the proposed stability conditions over existing ones.
\end{abstract}

\section{Introduction} \label{sec:intro}

As the fields of cyber-physical systems (CPS) and networked control systems (NCS) became more and more prevalent in both research and practice in recent decades, there has also been a renewed interest in the study of sampled-data systems, and especially in systems under aperiodic sampling (see \cite{hetel2017recent} for an overview). Among the reasons for this development is that many challenges and concepts arising in these fields, e.g., packet dropouts, delays, or event-triggered strategies, yield aperiodic sampling. For the aforementioned scenarios, there is often an upper bound on the time span between sampling instants available, e.g., if the number of consecutive packet dropouts and the length of delays is limited \cite{xiong2007packet_loss} or if the inter-event times are upper bounded \cite{gleizer2021bisimulation}. Hence, the resulting feedback loops can be analyzed by abstracting them as an aperiodically sampled system with arbitrary, but upper bounded sampling intervals. The maximum sampling interval (MSI) is a vital quantity in this respect, since it places a bound on how long the sampling intervals can become without sacrificing stability.

A great majority of existing works on aperiodically sampled systems is focused on continuous-time systems (cf. \cite{hetel2017recent}). However, a discrete-time formulation often arises naturally from the typical applications of aperiodically sampled systems, namely CPS and NCS, where sensing, control and actuation is taken care of by digital devices. A discrete-time setup was considered for instance in \cite{xiong2007packet_loss,hetel2011discrete} in the framework of switched systems, or in \cite{seuret2018wirtinger} using time-delay systems and Lyapunov-Krasovskii functionals for analysis.

All of the aforementioned works rely on accurate models in order to be effective, although such models can be challenging to obtain via first principles. Measured data of the system, on the other hand, are typically easy to acquire. One approach to leverage this fact is to estimate a model from the given data \cite{ljung1987system} and subsequently, to plug it into the model-based stability conditions \cite{hetel2017recent,xiong2007packet_loss,hetel2011discrete,seuret2018wirtinger}. Two shortcomings of this approach are on the one hand that it is often challenging to provide guarantees for the accuracy of a model estimated from finite and noisy data \cite{matni2019self}, and on the other that methods based on set membership estimation \cite{milanese1991optimal,belforte1990parameter} quickly grow in complexity with increasing system dimension.

As an alternative to this two-step procedure, there has recently emerged a stream of research which aims at giving system-theoretic guarantees directly from measured data \cite{willems2005note}. In this context, also data-driven approaches to analyze time-delay \cite{rueda2021delay} and aperiodially sampled systems \cite{berberich2021aper_samp,wildhagen2021dataMSI} were developed, guaranteeing closed-loop stability despite noisy measurements. In \cite{wildhagen2021dataMSI}, the considered discrete-time aperiodically sampled system was written as an interconnection of a linear time-invariant (LTI) system and a so-called delay operator. Then, the $\ell_2$ gain of this operator was bounded, and model-based as well as data-driven stability conditions were derived using robust control theory. This procedure closely resembles the robust input/output approach to aperiodically sampled systems, which is well-explored in continuous time \cite{mirkin2007some,fujioka2009IQC}, but has not received much attention in discrete time so far. Nonetheless, a formulation of the problem in discrete time is not only meaningful for CPS and NCS, but it also emerges naturally for data-based approaches, since data can only be measured at discrete time instants in any practical scenario.

In this paper, we refine the robust input/output approach to discrete-time aperiodically sampled systems by providing a more thorough analysis of the delay operator. In particular, we verify that in addition to the $\ell_2$ gain property, the delay operator is input-feedforward passive. Furthermore, we compute its $\ell_2$ gain exactly and we prove that this value is strictly smaller than the bound given in \cite{wildhagen2021dataMSI}. Based on this, we state improved stability conditions for the aperiodically sampled system, which guarantee stability for any sampling pattern satisfying a known upper bound on the sampling interval. First, we present criteria using full model knowledge and second, we provide data-driven conditions to analyze stability using only a finite-length, noise-corrupted state-input trajectory of the otherwise unknown system. In two examples, we demonstrate that the proposed stability conditions indeed yield better results than existing ones.

The remainder of this paper is organized as follows. In Section \ref{sec:setup}, we present the setup and problem statement. In Section \ref{sec:IO}, we explain the main idea of the robust input/output approach and derive input-feedforward passivity and the $\ell_2$ gain of the delay operator. In Section \ref{sec:stab}, we state the model-based and data-driven stability conditions and finally, we illustrate our results in Section \ref{sec:ex} with two examples.

\emph{Notation:} Let $\N$ be the set of natural numbers, $\N_0\coloneqq \N\cup \{0\}$ and $\N_{[a,b]}\coloneqq\N_0\cap[a,b]$, $\N_{\ge a} \coloneqq \N_0\cap[a,\infty)$, $a,b\in\N_0$. We denote by $I$ the identity matrix and by $0$ the zero matrix of appropriate dimension. Let $A\in\R^{n\times n}$ be a real matrix. We write $A\succ0$ $(A\succeq 0)$ if $A$ is symmetric and positive (semi-)definite, and we denote negative (semi-)definiteness similarly. Let $\sigma_\text{max}(A)$ ($\lambda_\text{max}(A)$) denote the maximum singular (eigen-) value of $A$. We write $\lVert v\rVert_2$ for the 2-norm of a vector $v\in\R^n$, $\lVert A\rVert_2$ for the induced 2-norm of $A$, and $\lVert A\rVert_F$ for the Frobenius norm of $A$. The Hermitian transpose of a complex matrix $B\in\mathbb{C}^{n\times m}$ is denoted by $B^*$. The Kronecker product of two matrices $C\in\R^{n\times m}$ and $D\in\R^{p\times r}$ is denoted by $C\otimes D$. We write $\star$ if an element in a matrix can be inferred from symmetry. We denote by $\ell_2$ the space of square integrable signals and by $\ell_{2e}$ the extended $\ell_2$ space. For some $T\in\N_0$, we denote by $\cdot_T:\ell_{2e}\to\ell_{2e}$ the truncation operator, which assigns to a signal $y\in\ell_{2e}$ the signal $y_T$ which satisfies $y_T(t)=y(t)$ for all $t\in\N_{[0,T]}$ and $y_T(t)=0$ for all $t\in\N_{\ge T+1}$. We write $\lVert x\rVert_{\ell_2}$ for the $\ell_2$ norm of a signal $x\in\ell_{2}$ and $\lVert \Delta\rVert_{\ell_2}\coloneqq\inf\{\gamma\;\vert\; \lVert\Delta(y)_T\rVert_{\ell_2}\le\gamma\lVert y_T\rVert_{\ell_2}, \; y\in\ell_{2e}, \; T\in\N_0\}$ for the $\ell_2$ gain of an operator $\Delta: \ell_{2e}\to\ell_{2e}$.

e use the well-established concept of integral quadratic constraints (IQCs) (see \cite{megretski1997system,kao2012discrete_delayed_IQC,hu2017discreteIQC,scherer2021dissipativity}) in order to describe input/output properties of the delay operator. A definition of a so-called hard static IQC, which we will use throughout this paper, is given below.
\begin{definition} \label{def:IQC}
	A bounded, causal operator $\Delta:\ell_{2e}^p\to\ell_{2e}^q$, $y\mapsto e$ satisfies the hard static IQC defined by a multiplier $\Pi\in\R^{(p+q) \times (p+q)}$ if for all $y\in\ell_2^p$, $e=\Delta(y)$, it holds that
		\begin{equation} \label{eq:def_IQC}
		\sum_{t=0}^{T}
		\begin{bmatrix} y(t) \\ e(t) \end{bmatrix}^\top
		\Pi \begin{bmatrix} y(t) \\ e(t) \end{bmatrix} \ge 0 , \quad\forall T\in\N_0.
		\end{equation}
\end{definition}
We write $\Delta\in\text{IQC}(\Pi)$ if $\Delta$ satisfies the hard static IQC \eqref{eq:def_IQC} in the sense of Definition \ref{def:IQC}.

\section{Setup and Problem Statement} \label{sec:setup}

In this paper, we consider a discrete-time LTI system
\begin{equation} \label{eq:system}
	x(t+1) = A_\text{tr} x(t) + B_\text{tr} u(t), \; x(0)=x_0\in\R^n
\end{equation}
with state $x(t)\in\R^n$, input $u(t)\in\R^m$ and time $t\in\N_0$. In closed-loop operation, \eqref{eq:system} is sampled and controlled at aperiodic \textit{sampling instants} $t_k\in\N_0$, $k\in\N_0$, where
\begin{equation*}
t_0 = 0, \quad t_{k+1}-t_k \ge 1,
\end{equation*}
meaning that the \textit{sampling interval} $h_k\coloneqq t_{k+1}-t_k$ is time-varying. The only a priori knowledge that we have about the sampling intervals is that they are upper bounded by a constant $\overline{h}\in\N$, i.e., $h_k\in\N_{[1,\overline{h}]}$. The sampled plant state $x(t_k)$ is available to the controller, from which it computes the control inputs via a linear state-feedback law $u(t_k)=K x(t_k)$, $K\in\R^{m\times n}$. The input is then held constant in between sampling instants $u(t)=u(t_k), \; t\in \N_{[t_k,t_{k+1}-1]}$. We summarize the closed-loop system under aperiodic sampling as
\begin{align}
x(t\hspace{-1pt}+\hspace{-1pt}1) &= A_\text{tr} x(t) \hspace{-1pt}+\hspace{-1pt} B_\text{tr} K x(t_k), \: \forall t\in\N_{[t_k,t_{k+1}\hspace{-0.5pt}-\hspace{-0.5pt}1]}, \: \forall k\in\N_{0} \nonumber \\
t_{k+1} &= t_k + h_k, \; h_k\in\N_{[1,\overline{h}]}, \; \forall k\in\N_{0} \label{eq:system_aper_sampled} \\ 
t_0 &= 0, \quad x(0) = x_0. \nonumber
\end{align}

Since as discussed in the introduction, many problems arising in CPS and NCS can be abstracted by an aperiodically sampled system with known upper bound on the sampling interval, we consider stability analysis of \eqref{eq:system_aper_sampled} for a given $\overline{h}$. In addition, we are interested in a preferably tight lower bound on the MSI, i.e., in the maximum $\overline{h}$ such that stability of \eqref{eq:system_aper_sampled} can be guaranteed. This quantity will be denoted by $\overline{h}_\text{MSI}$. Naturally, once tractable conditions for the former problem are found, the latter can be solved via a linear search over $\overline{h}$ \cite[Chapter 6]{knuth1997art} or an exponential search \cite{bentley1976almost}. For this reason, we concentrate on the former problem for the technical results.

We aim to solve these problems both in case of full model knowledge, i.e., when the true matrices $A_\text{tr}$ and $B_\text{tr}$ are known with certainty, and in the pure data-driven case. Therein, we have access only to a finite-length state-input trajectory of the system, which is additionally subject to noise, whereas $A_\text{tr}$ and $B_\text{tr}$ are entirely unknown. Although not presented here, we expect that extending our results to the case when both data \emph{and} some prior knowledge on $A_\text{tr}$ and $B_\text{tr}$ are available is straightforward using the framework in \cite{berberich2020combining}.

\section{Robust input/output approach} \label{sec:IO}

\subsection{Main Idea} \label{sec:IO_mainidea}

First, we rewrite the closed-loop aperiodically sampled system \eqref{eq:system_aper_sampled} as a time-delay system
\begin{equation*} \label{eq:system_delayed}
\begin{aligned}
x(t+1) &= (A_\text{tr} + B_\text{tr}K) x(t) + B_\text{tr}K (x(t_k)-x(t)) \\
&= (A_\text{tr} + B_\text{tr}K) x(t) + B_\text{tr}K (x(t-\tau(t))-x(t))
\end{aligned}
\end{equation*}
for $t\in\N_{[t_k,t_{k+1}-1]}$, $k\in\N_0$. We define
\begin{equation*}
\tau(t)\coloneqq t-t_k, \; t\in\N_{[t_k,t_{k+1}-1]}, \; k\in\N_0,
\end{equation*}
which is the amount of time by which the feedback information is delayed at time $t$. Note that the \textit{delay sequence} $\{\tau(t)\}$ has a ``sawtooth shape'' as illustrated in Figure \ref{fig:delay_staircase}, i.e., it is reset to zero at sampling instants and is increased by one at all the other time instants.

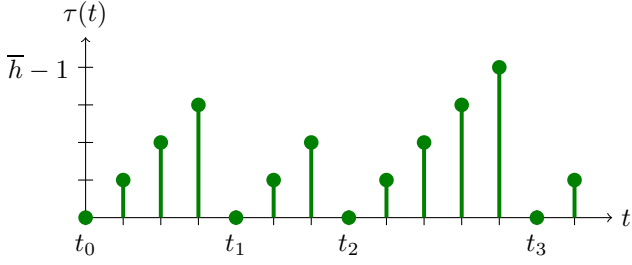
\begin{figure}
	\begin{tikzpicture}
	\draw[->] (0,0) -- (7,0) node[right] {$t$};
	\draw[->] (0,0) -- (0,2.4) node[above] {$\tau(t)$};
	\draw (-0.1,2) node[left] {$\overline{h}-1$} -- (0.1,2);
	\draw (0,-0.1) node[below] {$t_0$} -- (0,0.1);
	\draw (2,-0.1) node[below] {$t_1$} -- (2,0.1);
	\draw (3.5,-0.1) node[below] {$t_2$} -- (3.5,0.1);
	\draw (6,-0.1) node[below] {$t_3$} -- (6,0.1);
	
	\foreach \x in {0.5,1,1.5,2.5,3,4,4.5,5,5.5,6.5}
	\draw (\x,-0.1) -- (\x,0.1);
	
	\foreach \y in {0.5,1,1.5}
	\draw (-0.1,\y) -- (0.1,\y);
	
	\draw[color=istgreen,line width = 1.5pt] plot[ycomb,mark=*] coordinates {(0,0) (0.5,0.5) (1,1) (1.5,1.5) (2,0)  (2.5,0.5)(3,1) (3.5,0) (4,0.5) (4.5,1) (5,1.5) (5.5,2) (6,0) (6.5,0.5) };
	\end{tikzpicture}
	\caption{Staircase shape of $\tau(t)$.}
	\label{fig:delay_staircase}
\end{figure}

The quantity $e(t) \coloneqq x(t-\tau(t)) - x(t)$ can be interpreted as the ``error'' resulting from the aperiodic sampling. Note that it can be represented by a telescopic sum
\begin{align*}
e(t) & = x(t-\tau(t))-x(t) =  x(t-\tau(t)) \nonumber \\
&-x(t-\tau(t)+1)+x(t-\tau(t)+1)-\ldots-x(t) \nonumber \\
&= \sum_{i=t-\tau(t)}^{t-1} x(i)-x(i+1).
\end{align*}
We introduce the artificial output $y(t) \coloneqq x(t)-x(t+1)$ and define the \textit{delay operator} $\Delta:\ell_{2e}^n\to\ell_{2e}^n$, $e=\Delta y$, as
\begin{equation*} \label{eq:def_delta}
e(t) = (\Delta y)(t) \coloneqq \sum_{i= t-\tau(t)}^{t-1}  y(i), \; t\in\N_{[t_k,t_{k+1}-1]}, \; k\in\N_0,
\end{equation*}
in order to write the aperiodically sampled system \eqref{eq:system_aper_sampled} as an interconnection of an LTI system and the delay operator
\begin{subequations} \label{eq:fb_interconnection_ss}
	\begin{align}
	\begin{bmatrix}
	x(t+1) \\ y(t)
	\end{bmatrix}&=
	\begin{bmatrix}
	A_\text{tr}+B_\text{tr}K& B_\text{tr}K \\
	I-A_\text{tr}-B_\text{tr}K&-B_\text{tr}K
	\end{bmatrix}
	\begin{bmatrix}
	x(t)\\ e(t)\end{bmatrix}, \label{eq:fb_interc_syst} \\[3pt]
	e(t)&=(\Delta y)(t). \label{eq:fb_interc_delay}
	\end{align}
\end{subequations}

The main idea is now to analyze stability of the feedback interconnection \eqref{eq:fb_interconnection_ss} using robust control theory, where the sampling-induced error $e$ is comprehended as a disturbance acting on the ``nominal'' LTI system \eqref{eq:fb_interc_syst}. To this end, the delay operator $\Delta$ is embedded into a class of uncertainties by describing its input/output behavior using IQCs. In this respect, it is important to note that the tighter the uncertainty description of $\Delta$ is, the less conservative the resulting stability conditions will be. In the following, we will verify that the delay operator satisfies input-feedforward passivity. This is novel for the discrete-time setup and resembles the results of \cite{fujioka2009IQC}, wherein a similar property for the continuous-time equivalent of $\Delta$ was established. Furthermore, we assert that the delay operator's $\ell_2$ gain is determined by the maximum eigenvalue of a certain matrix, whereby we improve (i.e., decrease) the $\ell_2$ gain estimate from \cite{wildhagen2021dataMSI}. Subsequently, we express both conditions in terms of a hard static IQC.

\subsection{Input/Output Properties of the Delay Operator} \label{sec:IO_delay_operator}

In this subsection, we verify input-feedforward passivity and compute the $\ell_2$ gain of the delay operator, which are the main technical contributions of this paper. Input-feedforward passivity is established in the following result.
\begin{lemma} \label{lem:passive}
	For all $\mathcal{Y}=\mathcal{Y}^\top \succeq 0$ and for all $y\in\ell_{2}^n$, $e=\Delta y$, the delay operator $\Delta$ satisfies
	\begin{equation} \label{eq:passive}
	\sum_{t=0}^{T} \left(y(t)^\top \mathcal{Y} e(t) + \frac{1}{2}y(t)^\top \mathcal{Y} y(t)\right) \ge 0, \quad\forall T\in\N_0.
	\end{equation}
\end{lemma}
\noindent The proof can be found in Appendix \ref{app:proof_lem_passive}.

\begin{remark}
	By setting $\mathcal{Y}\coloneqq I$ in \eqref{eq:passive}, we retain the classical condition for input-feedforward passivity of $\Delta$ (cf. \cite{khalil2002nonlinear,kottenstette2014relationships}). We consider a general multiplier $\mathcal{Y}\succeq 0$, since this will enable us to formulate less conservative stability conditions later.
\end{remark}

\begin{remark}
	It is clear that \eqref{eq:passive} holds as well if the factor $\frac{1}{2}$ is replaced by an arbitrary $c\ge\frac{1}{2}$. From the proof of Lemma \ref{lem:passive}, it is easy to see that $\frac{1}{2}$ is indeed the smallest possible factor such that input-feedforward passivity holds.
\end{remark}

\begin{remark}
	Interestingly, the continuous-time equivalent of the delay operator $\Delta$ was proven to satisfy passivity without the feedforward term \cite{fujioka2009IQC}, which is a stronger property than \eqref{eq:passive}. Nonetheless, input-feedforward passivity will be useful for stability analysis as we will see later.
\end{remark}

To characterize the $\ell_2$ gain of $\Delta$, we define the matrix
\begin{equation*}
E_{\overline{h}} \coloneqq \begin{bmatrix}
0 & \cdots & \cdots & \cdots & \cdots & 0 \\
\vdots & 1 & \cdots & \cdots & \cdots & 1 \\
\vdots & \vdots & 2 & \cdots & \cdots & 2 \\
\vdots & \vdots & \vdots & \ddots &  & \vdots \\
\vdots & \vdots & \vdots & & \overline{h}-2 & \overline{h}-2 \\
0 & 1 & 2 & \cdots & \overline{h}-2 & \overline{h}-1 \\
\end{bmatrix}\in\R^{\overline{h}\times\overline{h}}.
\end{equation*}

\begin{lemma} \label{lem:L2_gain}
	The delay operator $\Delta$ has $\ell_2$ gain $\lVert \Delta\rVert_{\ell_2}=\sqrt{\lambda_\text{max}(E_{\overline{h}})}$, i.e., for all $y\in\ell_{2}^n$, $e=\Delta y$, we have
	\begin{equation*}
	\sum_{t=0}^{T} e(t)^\top e(t) \le \lambda_\text{max}(E_{\overline{h}}) \sum_{t=0}^{T} y(t)^\top y(t), \quad\forall T\in\N_0.
	\end{equation*}
\end{lemma}
\noindent The proof can be found in Appendix \ref{app:proof_lem_L2}.

Note that $\lambda_\text{max}(E_{\overline{h}})$, as given in Lemma \ref{lem:L2_gain}, is a tight bound on the (squared) $\ell_2$ gain of $\Delta$. The following result shows that it is strictly smaller than $\frac{\overline{h}}{2}(\overline{h}-1)$, which is the (squared) $\ell_{2}$ gain bound previously given in \cite[Lemma 4]{wildhagen2021dataMSI}.
\begin{proposition} \label{prop:guaranteed_imprv_l2_gain}
	It holds that
	\begin{equation} \label{eq:guaranteed_imprv_l2_gain}
	\lambda_\text{max}(E_{\overline{h}})\hspace{-1pt}\le\hspace{-1pt}\lVert E_{\overline{h}}\rVert_F\hspace{-1pt}=\hspace{-1pt}\sqrt{\frac{1}{6}(\overline{h}-1)\overline{h}(\overline{h}^2\hspace{-2pt}-\hspace{-0pt}\overline{h}+1)}\hspace{-1pt}\le\hspace{-1pt}\frac{\overline{h}}{2}(\overline{h}-1).
	\end{equation} Further, $\lambda_\text{max}(E_{\overline{h}})<\frac{\overline{h}}{2}(\overline{h}-1)$ for all $\overline{h}\in\N_{\ge 3}$.
\end{proposition}
\noindent The proof can be found in Appendix \ref{app:proof_prop_guaranteed_improvement}.

The middle term in \eqref{eq:guaranteed_imprv_l2_gain} comes from the fact that the maximum eigenvalue of a positive semi-definite matrix is upper bounded by its Frobenius norm. The latter, in contrast to the maximum eigenvalue itself, allows for a simple expression of the dependency on $\overline{h}$. Proposition \ref{prop:guaranteed_imprv_l2_gain} also reeals that the quotient of the Frobenius estimate and the estimate given in \cite{wildhagen2021dataMSI} converges to $\frac{\sqrt{2}}{\sqrt[4]{6}}\approx 0.9036$ as $\overline{h}\to\infty$. Numerically, we observed that the quotient of the estimate in Lemma \ref{lem:L2_gain} and the one given in \cite{wildhagen2021dataMSI} converges to approximately $0.9003$.

Both input-feedforward passivity in Lemma \ref{lem:passive} and the $\ell_2$ gain in Lemma \ref{lem:L2_gain} can be expressed by a hard static IQC.

\begin{corollary} \label{cor:combined_IQC}
	For any	$\mathcal{X}=\mathcal{X}^\top \succ 0$ and $\mathcal{Y}=\mathcal{Y}^\top \succeq 0$, it holds that $\Delta\in\text{IQC}(\Pi)$, where
	\begin{equation*}
	\Pi\coloneqq\begin{bmatrix} \lambda_{\text{max}}(E_{\overline{h}}) \mathcal{X}+\mathcal{Y} & \mathcal{Y} \\
	\mathcal{Y} & - \mathcal{X}
	\end{bmatrix}.
	\end{equation*}
\end{corollary}
\begin{proof}
	Lemma \ref{lem:passive} implies  $\Delta\in\text{IQC}(\Pi_{P})$, while Lemma \ref{lem:L2_gain} implies $\Delta\in\text{IQC}(\Pi_{\ell_2})$ (\cite[Corollary 7]{wildhagen2021dataMSI}), where
	\begin{equation*}
	\Pi_P \coloneqq \begin{bmatrix}	\mathcal{Y} & \mathcal{Y} \\ \mathcal{Y} & 0 \end{bmatrix} \text{ and }
	\Pi_{\ell_2} \coloneqq	\begin{bmatrix}	\lambda_{\text{max}}(E_{\overline{h}})\mathcal{X} & 0 \\ 0 &  -\mathcal{X} \end{bmatrix}.
	\end{equation*}
	The statement follows from the fact that $\Delta\in\text{IQC}(\nu_{\ell_2}\Pi_{\ell_2}+\nu_P \Pi_P)$ for all $\nu_{\ell_2},\nu_P\in\R_{\ge 0}$ (cf. \cite[Remark 2]{kao2012discrete_delayed_IQC}).
\end{proof}

\section{Stability criteria} \label{sec:stab}

\subsection{Model-based} \label{sec:stab_stab_model}

In this subsection, we assume that the true system matrices $A_\text{tr}$ and $B_\text{tr}$ are known. Having established a hard static IQC for the delay operator $\Delta$ in Corollary \ref{cor:combined_IQC}, one may now use existing results in robust control theory \cite{scherer2021dissipativity} to provide stability conditions for \eqref{eq:fb_interconnection_ss}. The following result follows directly from \cite[Corollary 11]{scherer2021dissipativity}.
\begin{theorem} \label{thm:stab_IO_ss}
	Suppose there exist matrices $\mathcal{S}=\mathcal{S}^\top\succ 0\in\R^{n\times n}$, $\mathcal{X} = \mathcal{X}^\top \succ 0\in\R^{n\times n}$ and $\mathcal{Y}=\mathcal{Y}^\top \succeq 0\in\R^{n\times n}$ such that \eqref{eq:stab_cond_IO_ss} is satisfied. Then the origin of \eqref{eq:fb_interconnection_ss} is asymptotically stable.
	\begin{figure*}
		\vspace{2pt}
		\begin{align}\label{eq:stab_cond_IO_ss}
		\left[
		\begin{array}{cc}
		A_\text{tr}+B_\text{tr}K&B_\text{tr}K\\I&0\\\hline
		I-A_\text{tr}-B_\text{tr}K&-B_\text{tr}K\\0&I
		\end{array}
		\right]^\top
		\left[
		\begin{array}{c|c}
		\begin{matrix} \mathcal{S} & 0 \\ 0 & -\mathcal{S} \end{matrix} & \begin{matrix}0 & 0 & \\0 & 0 \end{matrix} \\\hline
		\begin{matrix} 0 & 0 & \\0 & 0 \end{matrix}&\begin{matrix} \lambda_{\text{max}}(E_{\overline{h}})\mathcal{X}+\mathcal{Y} & \mathcal{Y} \\ \mathcal{Y} & -\mathcal{X} \end{matrix}
		\end{array}
		\right]
		\left[
		\begin{array}{cc}
		A_\text{tr}+B_\text{tr}K&B_\text{tr}K\\I&0\\\hline
		I-A_\text{tr}-B_\text{tr}K&-B_\text{tr}K\\0&I
		\end{array}
		\right]
		\prec0.
		\end{align}
		\noindent\makebox[\linewidth]{\rule{\textwidth}{0.4pt}}
	\end{figure*}
\end{theorem}

Condition \eqref{eq:stab_cond_IO_ss} is a linear matrix inequality (LMI) in all variables, such that an efficient search for suitable multipliers $\mathcal{X}$ and $\mathcal{Y}$ is possible via a semi-definite program (SDP). It is also possible to translate the finite-dimensional condition \eqref{eq:stab_cond_IO_ss} into the frequency domain using the KYP lemma \cite{rantzer1996kyp}. Let us denote the transfer function of \eqref{eq:fb_interc_syst} by
\begin{equation*}
G(z) = (I-A_\text{tr}-B_\text{tr}K)(zI-(A_\text{tr}+B_\text{tr}K))^{-1}B_\text{tr}K-B_\text{tr}K.
\end{equation*}
\begin{corollary} \label{cor:stab_IO_freq}
	Suppose $A_\text{tr}+B_\text{tr}K$ is Schur and there exist matrices $\mathcal{X} = \mathcal{X}^\top \succ 0\in\R^{n\times n}$ and $\mathcal{Y}=\mathcal{Y}^\top \succeq 0\in\R^{n\times n}$ such that
	\begin{equation} \label{eq:stab_cond_IO_freq}
	\begin{bmatrix} G(e^{j\omega}) \\ I \end{bmatrix}^*
	\begin{bmatrix} \lambda_{\text{max}}(E_{\overline{h}}) \mathcal{X}+\mathcal{Y} & \mathcal{Y} \\
	\mathcal{Y} & - \mathcal{X}
	\end{bmatrix}
	\begin{bmatrix} G(e^{j\omega}) \\ I \end{bmatrix}
	\prec 0
	\end{equation}
	is satisfied for all $\lvert \omega\rvert\le\pi$. Then, the origin of \eqref{eq:fb_interconnection_ss} is asymptotically stable.
\end{corollary}

\subsection{Data-driven} \label{sec:stab_data}

In this subsection, we assume that the true system matrices $A_\text{tr}$ and $B_\text{tr}$ are unknown. Instead, state-input data $\{x(t)\}_{t=0}^{N}$, $\{u(t)\}_{t=0}^{N-1}$, $N\in\N$, of the perturbed system
\begin{equation*} \label{eq:system_data_perturbed}
x(t+1) = A_\text{tr} x(t) + B_\text{tr} u(t) + B_d d(t)
\end{equation*}
are available, where $d(t)\in\R^{n_d}$ is an unknown disturbance and $B_d$ is a known matrix. The disturbance may account for a noise-corrupted experiment and $B_d$ can be used to incorporate prior knowledge about this disturbance, e.g., if it is certain that it only affects a subset of the states. Although these measurements are taken at each of the time instants (and not, e.g., at aperiodic sampling instants), it is not restrictive to assume that such data are available. This is because the state- and input measurements can be buffered at the sensor and actuator, respectively, and extracted once the experiment is finished. Furthermore, note that aperiodic sampling is only relevant for closed-loop operation, while the required data can be obtained in an open-loop experiment.

The particular disturbance sequence $\{\hat{d}(t)\}_{t=0}^{N-1}$ that affected the measured data is unknown, but assumed to satisfy a known bound defined via $\hat{D}\coloneqq \begin{bmatrix} \hat{d}(0) & \cdots & \hat{d}(N-1)\end{bmatrix}$.
\begin{assumption} \label{ass:disturbance_bound}
	The disturbance satisfies $\hat{D}\in\mathcal{D}$, where
	\begin{align*}
	\mathcal{D}\coloneqq\Big\{D\in\mathbb{R}^{n_d\times N}\Bigm|
	\begin{bmatrix}D^\top\\I\end{bmatrix}^\top
	\begin{bmatrix}Q_d&S_d\\S_d^\top&R_d\end{bmatrix}
	\begin{bmatrix}D^\top\\I\end{bmatrix}\succeq0\Big\},
	\end{align*}
	for known $Q_d\prec 0\in\R^{N\times N},S_d\in\R^{N\times n_d},R_d\in\R^{n_d\times n_d}$.
\end{assumption}

Such a disturbance bound may encompass a number of practically relevant scenarios, e.g., a component-wise norm bound or a norm bound on the sequence $\hat{D}$ (cf. \cite{berberich2020combining,waarde2020from}). The following assumption on $B_d$ is essentially without loss of generality (cf. \cite{berberich2020combining}).
\begin{assumption} \label{ass:Bd}
	The matrix $B_d$ has full column rank.
\end{assumption}

Due to the unknown disturbance, an entire range of matrices $A,B$ could explain the recorded data. If we arrange the data according to
\begin{align*}
X^+&=\begin{bmatrix}x(1)&x(2)&\dots&x(N)\end{bmatrix},\\
X&=\begin{bmatrix}x(0)&x(1)&\dots&x(N-1)\end{bmatrix},\\
U&=\begin{bmatrix}u(0)&u(1)&\dots&u(N-1)\end{bmatrix},
\end{align*}
we can define the set of matrices compatible with the data and the disturbance bound as
\begin{equation*}
\Sigma_{AB}=\{\begin{bmatrix} A & B \end{bmatrix}\mid X^+=AX+BU+B_dD,D\in\mathcal{D}\}.
\end{equation*}
Now, if we want to guarantee stability of \eqref{eq:fb_interconnection_ss} despite inexact knowledge of the system matrices, we need to ensure that the model-based stability conditions hold for all $\begin{bmatrix} A & B \end{bmatrix}\in\Sigma_{AB}$. To achieve this, we make use of a data-driven parametrization of the matrices contained in $\Sigma_{AB}$. It follows directly from \cite[Lemma 4]{waarde2020from} or \cite[Lemma 2]{berberich2020combining} that if we define
\begin{equation*}
P_{AB}=\begin{bmatrix} Q_{AB} & S_{AB} \\ S_{AB}^\top & R_{AB} \end{bmatrix}
\coloneqq\left[\begin{array}{cc}
-X&0\\-U&0\\\hline X^+&B_d
\end{array}\right]
\begin{bmatrix}Q_d&S_d\\S_d^\top&R_d\end{bmatrix}
\star^\top,
\end{equation*}
then the set $\Sigma_{AB}$ can be expressed as
\begin{equation} \label{eq:param_IO}
\Sigma_{AB}=\Big\{\begin{bmatrix} A & B \end{bmatrix}\Bigm|
\begin{bmatrix}A^\top\\B^\top\\I\end{bmatrix}^\top
P_{AB}
\begin{bmatrix}A^\top\\B^\top\\I\end{bmatrix}\succeq0\Big\}.
\end{equation}
With this, we have represented the set $\Sigma_{AB}$ by a quadratic matrix inequality (QMI) in the variables $\begin{bmatrix} A & B \end{bmatrix}$. We make the following technical assumption on the matrix $P_{AB}$ involved in the data-driven parametrization \eqref{eq:param_IO}.
\begin{assumption} \label{ass:data_inv}
	The matrix $P_{AB}$ is invertible and has exactly $n_d$ positive eigenvalues.
\end{assumption}
This is typically satisfied in practice if the data are sufficiently rich and $B_d$ is (chosen to be) invertible, implying that $n_d=n$ (cf. \cite{berberich2021aper_samp}).

To be able to state data-driven stability conditions, we rewrite the interconnection of LTI system and delay operator \eqref{eq:fb_interconnection_ss} as a linear fractional transformation (LFT)
\begin{subequations}\label{eq:system_IO_uncertain}
	\begin{align}
	\left[\begin{array}{c}
	x(t+1)\\\hline y(t)\\z(t)
	\end{array}\right]&=
	\left[\begin{array}{c|cc}
	0&0&I\\\hline
	I&0&-I\\
	\begin{bmatrix}I\\K\end{bmatrix}&\begin{bmatrix}0\\K\end{bmatrix}&0
	\end{array}\right]
	\left[\begin{array}{c}
	x(t)\\\hline e(t)\\w(t)\end{array}\right],\\
	e(t)&=(\Delta y)(t),\\
	w(t)&=(\begin{bmatrix} A & B \end{bmatrix}z)(t),
	\end{align}
\end{subequations}
with \emph{two} uncertainty channels $y\mapsto e$ and $z\mapsto w$. The first represents the delay operator and fulfills $\Delta\in\text{IQC}(\Pi)$, while the second describes the uncertain system matrices due to the disturbance in the measured data and satisfies $\begin{bmatrix} A & B \end{bmatrix}\in\Sigma_{AB}$. Using the S-procedure and the QMI representation \eqref{eq:param_IO} of $\Sigma_{AB}$, we can state a data-driven criterion for stability of \eqref{eq:system_IO_uncertain}. As the true system matrices $A_\text{tr}$ and $B_\text{tr}$ are contained in $\Sigma_{AB}$, this directly implies stability of \eqref{eq:fb_interconnection_ss} and hence \eqref{eq:system_aper_sampled}.

\begin{theorem} \label{thm:stab_IO_data}
	Suppose Assumptions \ref{ass:disturbance_bound}, \ref{ass:Bd} and \ref{ass:data_inv} are satisfied. Furthermore, suppose there exist matrices $\mathcal{S}=\mathcal{S}^\top\succ 0\in\R^{n\times n}$, $\mathcal{X}=\mathcal{X}^\top\succ 0\in\R^{n\times n}$ and $\mathcal{Y}=\mathcal{Y}^\top\succeq 0\in\R^{n\times n}$ such that \eqref{eq:stab_cond_IO_data} is satisfied, where
	\begin{align*}
	\begin{bmatrix} \tilde{Q}_{AB} & \tilde{S}_{AB} \\ \tilde{S}_{AB}^\top & \tilde{R}_{AB} \end{bmatrix}\coloneqq\begin{bmatrix} Q_{AB} & S_{AB} \\ S_{AB}^\top & R_{AB} \end{bmatrix}^{-1} = P_{AB}^{-1}.
	\end{align*}
	Then, the origin of \eqref{eq:system_IO_uncertain} is asymptotically stable for any $\begin{bmatrix} A & B \end{bmatrix} \in\Sigma_{AB}$.
	\begin{figure*}
		\vspace{2pt}
		\begin{align}\label{eq:stab_cond_IO_data}
		\left[\begin{array}{ccc}
		0&0&I\\I&0&0\\\hline
		I&0&-I\\0&I&0\\\hline
		\begin{bmatrix} I \\ K \end{bmatrix}&\begin{bmatrix} 0 \\ K \end{bmatrix}&0\\0&0&I
		\end{array}
		\right]^\top
		\begingroup
		\renewcommand*{\arraystretch}{1.2}
		\left[
		\begin{array}{c|c|c}
		\begin{matrix}\mathcal{S}&0\\0&-\mathcal{S}\end{matrix}&
		\begin{matrix}0&0\\0&0\end{matrix}&\begin{matrix}0&0\\0&0\end{matrix}\\\hline
		\begin{matrix}0&0\\0&0\end{matrix}&\begin{matrix} \lambda_{\text{max}}(E_{\overline{h}})\mathcal{X}+\mathcal{Y} & \mathcal{Y} \\ \mathcal{Y} & -\mathcal{X} \end{matrix}&\begin{matrix}0&0\\0&0\end{matrix}\\\hline
		\begin{matrix}0&0\\0&0\end{matrix}&\begin{matrix}0&0\\0&0\end{matrix}&
		\begin{matrix} -\tilde{Q}_{AB}& \tilde{S}_{AB}\\ \tilde{S}_{AB}^\top& -\tilde{R}_{AB}\end{matrix}
		\end{array}
		\right]
		\endgroup
		\left[\begin{array}{ccc}
		0&0&I\\I&0&0\\\hline
		I&0&-I\\0&I&0\\\hline
		\begin{bmatrix} I \\ K \end{bmatrix}&\begin{bmatrix} 0 \\ K \end{bmatrix}&0\\0&0&I
		\end{array}
		\right]\prec0
		\end{align}
		\noindent\makebox[\linewidth]{\rule{\textwidth}{0.4pt}}
	\end{figure*}
\end{theorem}
\begin{proof}
	We apply the full-block S-procedure \cite{scherer2001lpv} to conclude that \eqref{eq:stab_cond_IO_data} implies that \eqref{eq:stab_cond_IO_ss} (wherein $A_\text{tr},B_\text{tr}$ are replaced by $A,B$) holds for all $\begin{bmatrix} A & B \end{bmatrix}$ which satisfy
	\begin{equation} \label{eq:param_alt}
	\begin{bmatrix}
	[A\;B] \\ I
	\end{bmatrix}^\top
	\begin{bmatrix} -\tilde{R}_{AB}& \tilde{S}_{AB}^\top\\ \tilde{S}_{AB}& -\tilde{Q}_{AB}\end{bmatrix}
	\begin{bmatrix}
	[A\;B] \\ I
	\end{bmatrix}
	\succeq 0.
	\end{equation}
	Applying the dualization lemma \cite[Lemma 4.9]{scherer2000linear} and using \eqref{eq:param_IO} reveals that $\begin{bmatrix} A & B \end{bmatrix}\in\Sigma_{AB}$ if and only if \eqref{eq:param_alt} holds (the required inertia properties hold by Assumption \ref{ass:data_inv}). Asymptotic stability for all $\begin{bmatrix} A & B \end{bmatrix}\in\Sigma_{AB}$ follows from Theorem \ref{thm:stab_IO_ss}.
\end{proof}

Just like it is the case for \eqref{eq:stab_cond_IO_ss}, Condition \eqref{eq:stab_cond_IO_data} is an LMI in all variables, such that a search for multipliers $\mathcal{X}$ and $\mathcal{Y}$ can be performed via an SDP.
\begin{remark}
	The improved estimate for the $\ell_2$ gain in Lemma \ref{lem:L2_gain} can also be exploited in the model-based and data-based stability conditions \cite[Theorems 8 and 9]{wildhagen2021dataMSI} by replacing $\frac{\overline{h}}{2}(\overline{h}-1)$ with $\lambda_{\text{max}}(E_{\overline{h}})$ therein.
\end{remark}
\begin{remark}
	If $\lambda_\text{max}(E_{\overline{h}})$ is too expensive to be evaluated numerically (which might be the case for large $\overline{h}$), it can be replaced by the Frobenius norm $\sqrt{\frac{1}{6}(\overline{h}-1)\overline{h}(\overline{h}^2-\overline{h}+1)}$ in Theorems \ref{thm:stab_IO_ss} and \ref{thm:stab_IO_data} and Corollary \ref{cor:stab_IO_freq}. In Subsection \ref{sec:IO_delay_operator}, we found that the Frobenius bound is in fact quite tight.
\end{remark}

\section{Examples} \label{sec:ex}

\subsection{Numerical analysis: example from \cite{fridman2010refined}}

In a first example, we analyze the system from \cite{fridman2010refined}, which was discretized with a base period of \SI{0.01}{\second}. The system matrices are then given by
\begin{equation*}
A_\text{tr} = \begin{bmatrix} 1 & 0.010000 \\ 0 & 0.999000 \end{bmatrix} \text{ and }
B_\text{tr} = \begin{bmatrix} 5\times 10^{-6} \\ 1.000\times 10^{-3} \end{bmatrix},
\end{equation*}
where we rounded off after 6 decimal places. As in \cite{fridman2010refined}, we consider the controller $K=-\begin{bmatrix} 3.75 & 11.5 \end{bmatrix}$. The numerical results were obtained using MatlabR2019b, YALMIP~\cite{YALMIP} and Mosek~\cite{MOSEK15}.

We consider a model-based stability analysis first. While using \cite[Theorem 8]{wildhagen2021dataMSI} yields an MSI estimate of $\overline{h}_\text{MSI}=122$, we obtain $\overline{h}_\text{MSI} = 136$ from Theorem \ref{thm:stab_IO_ss}, which is an improvement of approximately \SI{11.5}{\percent}. However, we note that if we fix $\mathcal{Y}=0$ in Theorem \ref{thm:stab_IO_ss}, we still obtain $\overline{h}_\text{MSI} = 136$. This suggests that the passivity multiplier has no effect on the MSI estimate for this particular system, and the observed improvement is only due to the refined $\ell_2$ gain. Nonetheless, in the second example in Subsection \ref{sec:ex_scalar}, we will demonstrate that incorporating passivity can indeed be useful.

Let us now assume that $A_\text{tr}$ and $B_\text{tr}$ are unknown and that we have $N=1000$ state-input measurements $\{x(t)\}_{t=0}^{N}$, $\{u(t)\}_{t=0}^{N-1}$ of the unknown system available. The data-generating input is sampled randomly from $[-10,10]$. We consider a perturbed measurement with a norm-bounded disturbance $\lVert\hat{d}(t)\rVert_{2}\le\overline{d}$ for some $\overline{d}\ge 0$. As discussed in \cite{berberich2020combining}, such a disturbance fulfills Assumption \ref{ass:disturbance_bound} with $Q_d=-I$, $S_d = 0$ and $R_d=\overline{d}^2NI$. Furthermore, we set $B_d = 0.01 I$ such that Assumption \ref{ass:Bd} is fulfilled and a certain $\overline{d}$ corresponds to an input-to-noise ratio of approximately $1/\overline{d}$.

For a data-driven estimation of the MSI, one may use either \cite[Theorem 9]{wildhagen2021dataMSI} or Theorem \ref{thm:stab_IO_data}. In Table \ref{tab:MSI}, the MSI estimates for both possibilities and different disturbance levels $\overline{d}$ are listed. Assumption \ref{ass:data_inv} was fulfilled in all of those cases. First of all, we observe that for small noise levels, both data-driven stability conditions yield an MSI estimate just as high as their model-based counterparts. Furthermore, we recognize that for all disturbance levels, the MSI estimates using Theorem \ref{thm:stab_IO_data} are higher than those using \cite[Theorem 9]{wildhagen2021dataMSI}. For $\overline{d}\ge 0.02$, none of the possibilities were feasible.

Finally, we note that when $\overline{d}$ is too small, $P_{AB}$ will become near singular which leads to numerical issues. In this example, this was the case for $\overline{d}\le 0.0005$. This is problematic only when using Theorem \ref{thm:stab_IO_data}, since the data-based stability conditions in \cite[Theorem 9]{wildhagen2021dataMSI} contain $P_{AB}$ directly instead of its inverse. Note that it would also be possible to replace the inverse of $\Pi_{\ell_{2}}$ in \cite[Theorem 9]{wildhagen2021dataMSI} by the inverse of $\Pi$ from Corollary \ref{cor:combined_IQC}, which would eliminate the numerical issues for small disturbances and would render Assumption \ref{ass:data_inv} needless. However, since the inverse of $\Pi$ depends nonlinearly on $\mathcal{X}$ and $\mathcal{Y}$ (and their inverses), a co-search for these multipliers via an SDP could not be performed in this case.

\begin{table}
	\centering
	\caption{Data-driven MSI estimates $\overline{h}_\text{MSI}$ for different approaches.}
	\begin{tabular}{c|ccccccc}
	$\overline{d}$ & 0.001 & 0.002 & 0.005 & 0.01 & 0.02  \\ \hline
	\cite[Theorem 9]{wildhagen2021dataMSI} & 122 & 122 & 121 & 115 & - \\
	Theorem \ref{thm:stab_IO_data} & 136 & 135 & 134 & 128 & - \\
	\end{tabular}
\label{tab:MSI}
\end{table}

\subsection{Frequency domain analysis: scalar systems} \label{sec:ex_scalar}

In a second example, we consider a scalar system $A_\text{tr}=a\in\R$ and $B_\text{tr}=b\in\R$,
$b\neq0$, and a fixed controller $K=1$. In this case, the transfer function of \eqref{eq:fb_interc_syst} is $G(z) = b\frac{1-z}{z-a-b}$.

With Corollary \ref{cor:stab_IO_freq} and by setting $\mathcal{X}=1$, we have stability if $a+b\in(-1,1)$ and there exists a $\mathcal{Y}\ge 0$ such that
\begin{equation*}
(\lambda_{\text{max}}(E_{\overline{h}})\hspace{-1pt}+\hspace{-1pt}\mathcal{Y})G^*(e^{j\omega})G(e^{j\omega}) \hspace{-1pt}+\hspace{-1pt}\mathcal{Y}(G(e^{j\omega})+G^*(e^{j\omega})\hspace{-1pt})-1\hspace{-1pt}<\hspace{-1pt}0
\end{equation*}
for all $|\omega|\le \pi$. This inequality is satisfied if
\begin{equation*}
(\lambda_{\text{max}}(E_{\overline{h}})+\mathcal{Y})G^*(e^{j\omega})G(e^{j\omega}) +\mathcal{Y}(G(e^{j\omega})+G^*(e^{j\omega}))\le 0.
\end{equation*}
We explicitly compute this term to see that this is the case if
\begin{equation} \label{eq:cond_ex_stability}
\frac{b^2}{1+a+b} \le \frac{b\mathcal{Y}}{\lambda_{\text{max}}(E_{\overline{h}})+\mathcal{Y}}
\end{equation}
holds. This inequality cannot be fulfilled if $b<0$, since $1+a+b>0$ and $\mathcal{Y}\ge 0$. Further, the right-hand side of \eqref{eq:cond_ex_stability} takes values in $[0,b)$ and we have $1+a+b<2$, from which it follows that $b<2$ must hold. As the right-hand side of \eqref{eq:cond_ex_stability} approaches $b$ as $\mathcal{Y}\to\infty$ regardless of $\overline{h}$, \eqref{eq:cond_ex_stability} can be fulfilled if $\frac{b}{1+a+b}<1$, i.e., $1+a> 0$. In summary, using Corollary \ref{cor:stab_IO_freq} we can guarantee stability for an \textit{arbitrary} $\overline{h}\in\N$ if $a\in(-1,1)$, $b\in(0,2)$ and $a+b\in(-1,1)$. In contrast, if we would not make use of input-feedforward passivity of the delay operator $\Delta$ (i.e., set $\mathcal{Y}=0$), \eqref{eq:cond_ex_stability} could never be fulfilled. Then, stability would hold if $\lambda_{\text{max}}(E_{\overline{h}})\lVert G\lVert_{\infty}-1<0$, where $\lVert G\lVert_{\infty}$ denotes the $H_\infty$ norm of $G$. This could only be fulfilled for a finite $\overline{h}$ as $\lVert G\lVert_{\infty}>0$. This shows that incorporating the passivity multiplier can bring a great benefit for estimation of the MSI, and that it may even enable to show stability of sampled-data systems with arbitrarily large sampling periods. For the above system, our results are also an improvement over the switched systems approach in \cite{xiong2007packet_loss,hetel2011discrete}, which becomes intractable for very large values of $\overline{h}$.

\bibliographystyle{IEEEtran}   
\bibliography{Literature}

\appendix

\section{Appendix}

\subsection{Proof of Lemma \ref{lem:passive}} \label{app:proof_lem_passive}

The claim follows directly if we verify that
\begin{equation*}
P(t)\coloneqq\sum_{i=t_k}^{t} y(i)^\top \mathcal{Y} (\Delta y)(i) \ge \sum_{i=t_k}^{t} -\frac{1}{2}y(i)^\top\mathcal{Y}y(i)
\end{equation*}
holds for all $t\in\N_{[t_k,t_{k+1}-1]}$, $k\in\N_0$. To this end, an alternative representation of the delay operator
\begin{equation} \label{eq:def_delta_alt}
(\Delta y)(i) = \sum_{j= t_k}^{i-1}  y(j), \; i\in\N_{[t_k,t_{k+1}-1]}, \; k\in\N_{0}
\end{equation}
will be useful. We start off by partitioning $P(t)$
\begin{equation} \label{eq:pass_delta_sh_1}
P(t)=\sum_{i=t_k}^{t-1} y(i)^\top \mathcal{Y} (\Delta y)(i) + y(t)^\top \mathcal{Y}(\Delta y)(t). 
\end{equation}
If $t=t_k$, $P(t)=0$ since $(\Delta y)(t_k)=0$. Thus, we focus on the case $t\in\N_{[t_k+1,t_{k+1}-1]}$ next. From the alternative definition of $\Delta$ in \eqref{eq:def_delta_alt}, we find
\begin{equation} \label{eq:difference_identity_delta_sh}
(\Delta y)(i+1)-(\Delta y)(i) = \sum_{j=t_k}^{i} y(j) - \sum_{j=t_k}^{i-1} y(j) = y(i) 
\end{equation}
for all $i\in[t_k,t_{k+1}-2]$. We plug \eqref{eq:difference_identity_delta_sh} into \eqref{eq:pass_delta_sh_1} to obtain
\begin{equation} \label{eq:pass_delta_sh_2}
\begin{aligned}
P(t)&=\sum_{i=t_k}^{t-1} \Big((\Delta y)(i+1)-(\Delta y)(i)\Big)^\top \mathcal{Y} (\Delta y)(i) \\
&+ y(t)^\top \mathcal{Y}(\Delta y)(t).
\end{aligned}
\end{equation}
Next, we analyze the summand in \eqref{eq:pass_delta_sh_2}. We add and subtract $\frac{1}{2}(\Delta y)(i+1)^\top\mathcal{Y}(\Delta y)(i+1)$ to rewrite it as
\begin{align*}
&\frac{1}{2}(\Delta y)(i+1)^\top\mathcal{Y}(\Delta y)(i+1)-\frac{1}{2}(\Delta y)(i)^\top\mathcal{Y}(\Delta y)(i) \\
&-\frac{1}{2}(\Delta y)(i+1)^\top\mathcal{Y}(\Delta y)(i+1)+(\Delta y)(i+1)^\top\mathcal{Y}(\Delta y)(i) \\
&-\frac{1}{2}(\Delta y)(i)^\top\mathcal{Y}(\Delta y)(i) \\
&= \frac{1}{2}(\Delta y)(i+1)^\top \mathcal{Y} (\Delta y)(i+1)-\frac{1}{2}(\Delta y)(i)^\top \mathcal{Y} (\Delta y)(i) \\
&-\frac{1}{2}\Big(\star\Big)^\top \mathcal{Y} \underbrace{\Big((\Delta y)(i+1)-(\Delta y)(i)\Big)}_{\stackrel{\eqref{eq:difference_identity_delta_sh}}{=}y(i)},
\end{align*}
where we used a completion of squares. As an intermediate step, let us evaluate the telescopic sum
\begin{align*}
\sum_{i=t_k}^{t-1} &\frac{1}{2}(\Delta y)(i+1)^\top \mathcal{Y} (\Delta y)(i+1)-\frac{1}{2}(\Delta y)(i)^\top \mathcal{Y} (\Delta y)(i) \nonumber \\
&= \frac{1}{2} (\Delta y)(t)^\top \mathcal{Y} (\Delta y)(t) -\frac{1}{2}(\Delta y)(t_k)^\top \mathcal{Y} \underbrace{(\Delta y)(t_k)}_{\stackrel{\eqref{eq:def_delta_alt}}{=}0}.
\end{align*}
With this, we can rewrite $P(t)$ as
\begin{align*}
P(t) &= \frac{1}{2} (\Delta y)(t)^\top \mathcal{Y} (\Delta y)(t) -\frac{1}{2} \sum_{i=t_k}^{t-1} y(i)^\top \mathcal{Y} y(i) \\
& + y(t)^\top \mathcal{Y}(\Delta y)(t).
\end{align*}
Finally, by adding and subtracting $\frac{1}{2} y(t)^\top \mathcal{Y} y(t)$ and using again a completion of squares, we obtain
\begin{align*}
&P(t) = -\frac{1}{2} \sum_{i=t_k}^{t} y(i)^\top \mathcal{Y} y(i)  +\frac{1}{2} \underbrace{\Big(\star\Big)^\top\hspace{-2pt} \mathcal{Y} \Big((\Delta y)(t)+y(t)\Big)}_{\ge 0 \text{ since } \mathcal{Y}\succeq 0}\hspace{-2pt}.
\end{align*}

\subsection{Proof of Lemma \ref{lem:L2_gain}} \label{app:proof_lem_L2}

Analogously to the proof of \cite[Lemma 4]{wildhagen2021dataMSI}, we find that the $\ell_2$ gain of $\Delta$ is equal to the $\ell_2$ gain of $D:\ell_{2}^{n}[0,\overline{h}-1]\to\ell_{2}^{n}[0,\overline{h}-1]$, $y\mapsto e$, $e(t)=(Dy)(t)\coloneqq\sum_{i=0}^{t-1}y(i)$, where $\ell_{2}[0,\overline{h}-1]$ denotes the space of signals of length $\overline{h}$. Next, we determine the $\ell_2$ gain of this operator.

We handle $D$ in the lifted domain. For a signal $g=\{g(0),$ $\ldots,g(\overline{h}-1)\}\in\ell_{2}^{n}[0,\overline{h}-1]$, the \textit{lifted signal} is defined as $\underline{g}\coloneqq\left\{\begin{bmatrix} g(0)^\top & \cdots & g(\overline{h}-1)^\top \end{bmatrix}^\top\right\}$.
Likewise, we may consider a lifted version of the operator $D$, according to $\underline{e}=\underline{D}_{\overline{h}}\: \underline{y}$ and
\begin{equation*}
\underline{D}_{\overline{h}} \coloneqq \begin{bmatrix} 0 & \cdots & 0 & 0 \\ I & \ddots & \vdots & \vdots \\ \vdots & \ddots & 0 & 0 \\ I & \cdots & I & 0 \end{bmatrix}\in\R^{\overline{h}\times\overline{h}}.
\end{equation*}
As lifting preserves the $\ell_2$ norm of signals $\lVert g\rVert_{\ell_2} = \lVert \underline{g}\rVert_{2}$ \cite{chen1995optimal}, the $\ell_2$ gain of $D$ is equal to the matrix 2-norm of $\underline{D}_{\overline{h}}$, i.e., $\lVert D\rVert_{\ell_2} = \lVert \underline{D}_{\overline{h}}\rVert_2$. To conclude the proof, we compute
\begin{align*}
\lVert\underline{D}_{\overline{h}}\rVert_2^2 &= \sigma_{\text{max}}(\underline{D}_{\overline{h}})^2 = \lambda_{\text{max}}(\underline{D}_{\overline{h}}\underline{D}_{\overline{h}}^\top) = \lambda_{\text{max}}(E_{\overline{h}}\otimes I) \\
&= \lambda_{\text{max}}(E_{\overline{h}})\lambda_{\text{max}}(I) = \lambda_{\text{max}}(E_{\overline{h}}).
\end{align*}

\subsection{Proof of Proposition \ref{prop:guaranteed_imprv_l2_gain}} \label{app:proof_prop_guaranteed_improvement}

It holds that
\begin{align*}
&\lambda_\text{max}(E_{\overline{h}}) \stackrel{E_{\overline{h}}\succeq 0}{=} \lVert E_{\overline{h}}\rVert_2 \le \lVert E_{\overline{h}}\rVert_F = \sqrt{\text{trace}(E_{\overline{h}}^\top E_{\overline{h}})} \\
&=\sqrt{\sum_{i=1}^{\overline{h}-1}\hspace{-1pt}\left(\hspace{-1pt}i^2(\overline{h}-i)+\sum_{j=1}^{i-1}j^2\hspace{-1pt}\right)\hspace{-1pt}} \hspace{-1pt}=\hspace{-1pt} \sqrt{\frac{1}{6}(\overline{h}-1)\overline{h}(\overline{h}^2\hspace{-1pt}-\hspace{-1pt}\overline{h}\hspace{-1pt}+\hspace{-1pt}1)}.
\end{align*}
Further, we have a look at $f(\overline{h})\coloneqq \frac{\lVert E_{\overline{h}}\rVert_F}{\frac{\overline{h}}{2}(\overline{h}-1)} = \frac{2}{\sqrt{6}}\sqrt{\frac{\overline{h}^2-\overline{h}+1}{\overline{h}^2-\overline{h}}}$. It is obvious that $f(\overline{h}+1)\le f(\overline{h})$ for all $\overline{h}\in\N_{\ge 2}$, since $\overline{h}\mapsto \overline{h}^2-\overline{h}$ is a strictly increasing function on $\N$. Noting that $f(3) = \sqrt{\frac{28}{36}}<1$ concludes the proof.

\end{document}